\documentclass[runningheads]{llncs}

\usepackage[T1]{fontenc}
\usepackage{amsmath,amssymb}
\usepackage{mathtools}
\usepackage{booktabs}
\usepackage{graphicx}
\usepackage{placeins}
\usepackage{xurl}
\usepackage[hidelinks]{hyperref}
\newcommand{\F}{\mathbb F}

\begin{document}

\title{Zeta-Transform Evaluation for Higher-Order Vanishing Key Recovery}
\titlerunning{Zeta-Transform Evaluation for Higher-Order Vanishing Key Recovery}

\author{
  Sunyeop Kim\inst{1,2}
  \and
  Insung Kim\inst{1}
}

\authorrunning{S. Kim and I. Kim}

\institute{
  Korea University, Republic of Korea\\
  \email{\{sunyeopkim,cmcom35\}@korea.ac.kr}
  \and
  Nanyang Technological University, Singapore\\
}

\maketitle

\begin{abstract}
Hemmert's key-recovery algorithm for Classic McEliece is based on
higher-order vanishing. It computes a basis of
$\ker(\widetilde\varphi_A^{(p)})$, where $A=H'''$ is the shortened
parity-check matrix used in the attack. For Classic McEliece parameters,
this kernel computation is the dominant cost of the attack. We show that
the sums defining $\widetilde\varphi_A^{(p)}$ can be evaluated, column
by column, as weighted upper zeta transforms on the Boolean lattice.
Since only selected levels of these transforms are required by
$\widetilde\varphi_A^{(p)}$, restricting their evaluation to the band
between level $p$ and the lowest required level yields exact evaluations
of both $\widetilde\varphi_A^{(p)}$ and its transpose. Using the
resulting truncated zeta-transform evaluation in the Wiedemann-based
kernel computation reduces the cost of the repeated matrix--vector
products without changing the overall key-recovery algorithm. The exact
cost depends on the weight distribution of the non-pivot columns of
$H'''$. We therefore consider two models: an all-one model, in which
every relevant binary coordinate is active, and a Bernoulli$(1/2)$
model, in which the coordinates are independently active with
probability $1/2$. For the five Classic McEliece parameter sets, our
method reduces the estimated key-recovery cost by $14.09$--$41.19$ bits
in the all-one model and by $7.25$--$22.48$ bits in the
Bernoulli$(1/2)$ model.
\end{abstract}

\section{The higher-order vanishing map}

We use the notation in~\cite[Definitions~4 and~5]{cryptoeprint:2026/1339}.
Let $q$ be a power of two, and let $k,n,p$ be positive integers with
$k\le n$ and $2\le p\le k$,
\[
 R_k\coloneqq \F_q[X_1,\ldots,X_k],
 \qquad [k]\coloneqq \{1,\ldots,k\}.
\]
Here $R_k^{(p)}$ denotes the $\F_q$-vector space of homogeneous
polynomials of degree $p$ in $R_k$, and $U_k^{(p)}\subset R_k^{(p)}$ is
the subspace spanned by its square-free monomials. For an integer $t$,
write
\[
 \binom{[k]}t\coloneqq \{S\subseteq[k]:|S|=t\}.
\]
Here $|S|$ denotes the cardinality of a finite set $S$.
Let $A\in\F_q^{k\times n}$ have rank $k$. Let
$\widetilde A\in\F_q^{k\times(n-k)}$ be obtained from the reduced
row-echelon form of $A$ by removing its pivot columns, and let
$\widetilde a_1,\ldots,\widetilde a_{n-k}$ denote the columns of
$\widetilde A$; write $\widetilde a_{j,i}$ for coordinate $i$ of column
$\widetilde a_j$. The following map is introduced
in~\cite[Definition~5]{cryptoeprint:2026/1339}:
\[
 \widetilde\varphi_A^{(p)}:
 U_k^{(p)}\longrightarrow
 \prod_{j=1}^{n-k}\F_q^{M_{k,p}},
 \qquad
 M_{k,p}\coloneqq
 \sum_{\ell=1}^{\lfloor\log_2p\rfloor}
 \binom{k}{p-2^\ell},
\]
whose entries are evaluations at $\widetilde a_j$ of the square-free
partial derivatives of orders $p-2^\ell$.

For
\[
 f=\sum_{S\in\binom{[k]}p}c_SX_S\in U_k^{(p)},
 \qquad X_S=\prod_{i\in S}X_i,
\]
the entry indexed by a column $\widetilde a_j$, an integer
$1\le\ell\le\lfloor\log_2p\rfloor$, and a set
$D=\{i_1<\cdots<i_{p-2^\ell}\}\in
\binom{[k]}{p-2^\ell}$ is
\begin{equation}
 \frac{\partial^{p-2^\ell}f}
      {\partial X_{i_1}\cdots\partial X_{i_{p-2^\ell}}}
      (\widetilde a_j)
 =
 \sum_{\substack{S\supseteq D\\|S|=p}}
 c_S\prod_{i\in S\setminus D}\widetilde a_{j,i}.
 \label{eq:hov-entry}
\end{equation}

We next specialize this notation to the key-recovery setting. Let
\[
 H\in\F_2^{mr\times n}
\]
be the public binary parity-check matrix of a binary Goppa code of
length $n$ over $\F_{2^m}$, defined by a Goppa polynomial of degree
$r$. For a permutation matrix $P\in\F_2^{n\times n}$, the procedure
in~\cite[Algorithm~4]{cryptoeprint:2026/1339} sets
\[
 H'\coloneqq HP,
 \qquad
 H''\coloneqq \operatorname{RREF}(H'),
\]
where $\operatorname{RREF}$ denotes reduced row-echelon form. Letting
$s$ be the shortening parameter returned together with $p$
in~\cite[Algorithm~5]{cryptoeprint:2026/1339}, $H'''$ is the matrix
obtained by deleting the first $s$ rows of $H''$. Under the independence
assumption in~\cite[Algorithm~4]{cryptoeprint:2026/1339}, $H'''$ has
rank $mr-s$. We now specialize the preceding general notation by setting
\begin{equation}
 A\coloneqq H''',
 \qquad k\coloneqq mr-s.
 \label{eq:specialization}
\end{equation}
The $s$ zero columns corresponding to the shortened positions contribute
no nonzero entries and may be omitted, leaving $n-mr$ contributing
columns. Let
\[
 d\coloneqq \dim\ker\!\left(\widetilde\varphi_{H'''}^{(p)}\right).
\]
The key-recovery algorithm computes a basis
\[
 \{v_1,\ldots,v_d\}
 \quad\text{of}\quad
 \ker\!\left(\widetilde\varphi_{H'''}^{(p)}\right).
\]

For Classic McEliece parameters, this kernel computation is the dominant
cost of the key recovery. We therefore focus on the Wiedemann-based
computation of $\ker(\widetilde\varphi_{H'''}^{(p)})$. For a detailed
description of the complete key-recovery attack based on higher-order
vanishing, we refer the reader to~\cite{cryptoeprint:2026/1339}. Our
modification changes only how $\widetilde\varphi_{H'''}^{(p)}$ and its
transpose are evaluated; all other steps of the key-recovery algorithm
remain unchanged.

\section{Truncated evaluation of a weighted zeta transform}

The sum in~\eqref{eq:hov-entry} has the form of a coordinate-weighted
upper zeta transform on the Boolean lattice
\[
 \bigl(2^{[k]},\subseteq\bigr),
\]
where $2^{[k]}$ is the set of all subsets of $[k]$ and the partial order
is set inclusion. We adapt the coordinatewise factorization underlying
the standard fast zeta-transform
algorithm~\cite{doi:10.1137/070683933,148425} to evaluate this weighted
sum.

Fix $j\in\{1,\ldots,n-k\}$ and write
$a=\widetilde a_j=(a_1,\ldots,a_k)$. Let
$z:2^{[k]}\to\F_q$ and define
\[
  (Z_a z)(T)
  \coloneqq
  \sum_{S\supseteq T}
  z(S)\prod_{i\in S\setminus T}a_i.
\]

The map $\widetilde\varphi_A^{(p)}$ does not require the values of this
transform on every level of the Boolean lattice. For each column
$\widetilde a_j$, only the levels $p-2^\ell$, where
$1\le\ell\le\lfloor\log_2p\rfloor$, occur in its output. Put
\[
  \ell_{\max}\coloneqq\lfloor\log_2p\rfloor,
  \qquad
  L\coloneqq p-2^{\ell_{\max}},
  \qquad
  \mathcal B_{k,p}
  \coloneqq
  \{T\subseteq[k]:L\le |T|\le p\}.
\]
The smallest required level is $L$. Since the coordinatewise
factorization passes through the intermediate levels between $p$ and
$L$, it is sufficient to maintain the band $\mathcal B_{k,p}$.

Starting with the coefficients of $f$ on level $p$, the following
theorem shows that the restriction of the weighted zeta transform to
the required levels can be evaluated entirely within this band and that
these levels give exactly the corresponding entries of
$\widetilde\varphi_A^{(p)}(f)$.

\begin{theorem}[Weighted zeta factorization]
Fix an ordering $u_1,\ldots,u_k$ of $[k]$, and maintain an array
$z(T)$ indexed by $T\in\mathcal B_{k,p}$. Initialize
\[
 z(T)=
 \begin{cases}
   c_T, & |T|=p,\\
   0,   & L\le |T|<p.
 \end{cases}
\]
For $t=1,\ldots,k$, set $i=u_t$ and, for every
$T\in\mathcal B_{k,p}$ satisfying $i\notin T$ and $|T|<p$, perform the
butterfly operation
\begin{equation}
  z(T)\gets z(T)+a_i z(T\cup\{i\}).
  \label{eq:butterfly}
\end{equation}
Then, for every $1\le\ell\le\ell_{\max}$ and
\[
 D=\{i_1<\cdots<i_{p-2^\ell}\}
   \in\binom{[k]}{p-2^\ell},
\]
we have
\[
 z(D)=
 \frac{\partial^{p-2^\ell}f}
      {\partial X_{i_1}\cdots\partial X_{i_{p-2^\ell}}}
      (\widetilde a_j),
\]
which is the corresponding entry of
$\widetilde\varphi_A^{(p)}(f)$.
\end{theorem}

\begin{proof}
For $0\le t\le k$, put
\[
 J_0\coloneqq\varnothing,
 \qquad
 J_t\coloneqq\{u_1,\ldots,u_t\},
\]
and let $z_t$ denote the state after the first $t$ coordinates have
been processed. We claim that, for every $T\in\mathcal B_{k,p}$,
\begin{equation}
 z_t(T)=
 \sum_{\substack{
       S\in\binom{[k]}p\\
       S\supseteq T\\
       S\setminus T\subseteq J_t}}
 c_S\prod_{h\in S\setminus T}a_h.
 \label{eq:zeta-invariant}
\end{equation}
For $t=0$, this is exactly the initialization. Suppose that the claim
holds for $t-1$, and write $i=u_t$. If $i\in T$ or $|T|=p$, no update
is performed, and the right-hand side of~\eqref{eq:zeta-invariant} is
unchanged. Otherwise, split the
contributing supersets according as they omit or contain $i$. The first
part is $z_{t-1}(T)$, while the second part is
$a_i z_{t-1}(T\cup\{i\})$. This is precisely the update
in~\eqref{eq:butterfly}, proving the invariant.

Every update with $L\le |T|<p$ reads the entry indexed by
$T\cup\{i\}$, whose cardinality is $|T|+1$. Hence the source also lies
in $\mathcal B_{k,p}$, and no entry outside the band is required.

Finally, $J_k=[k]$, so
\[
 z_k(T)=
 \sum_{\substack{S\supseteq T\\|S|=p}}
 c_S\prod_{h\in S\setminus T}a_h.
\]
For $|T|=p-2^\ell$, this is exactly~\eqref{eq:hov-entry}.
\qed
\end{proof}

The theorem gives a forward evaluation of the output block associated
with a single column $\widetilde a_j$. Since each butterfly operation is
linear, the transpose evaluation is obtained by processing the
coordinates in reverse order and replacing~\eqref{eq:butterfly} by
\[
 z(T\cup\{i\})
 \gets
 z(T\cup\{i\})+a_i z(T).
\]
The input is placed on the levels $p-2^\ell$, with all other band
entries initialized to zero, and the resulting level-$p$ entries form
the transpose output. Applying this procedure to each column block and
summing the results gives the transpose of the full map. Thus both
$\widetilde\varphi_A^{(p)}$ and its transpose can be evaluated exactly.
These evaluations replace the corresponding matrix--vector products in
the Wiedemann-based kernel computation~\cite{1057137} used in the
key-recovery attack~\cite{cryptoeprint:2026/1339}, without otherwise
changing the overall algorithm.

\section{Cost estimates}

We now use the notation
in~\cite[Proposition~13]{cryptoeprint:2026/1339}. Let $p$ and $s$ be the
parameters returned in~\cite[Algorithm~5]{cryptoeprint:2026/1339}. For
the matrix $H'''$ occurring
in~\cite[Algorithm~4]{cryptoeprint:2026/1339}, set the following two
dimensions:
\begin{equation}
 N_0\coloneqq (n-mr)
 \sum_{\ell=1}^{\lfloor\log_2p\rfloor}
 \binom{mr-s}{p-2^\ell},
 \qquad
 N_1\coloneqq \binom{mr-s}{p}.
 \label{eq:N0N1}
\end{equation}
After omitting the $s$ identically zero output blocks described above,
the resulting matrix representation of $\widetilde\varphi_{H'''}^{(p)}$
has $N_0$ rows and $N_1$ columns. Removing these zero blocks does not
change the kernel. The cost of the Wiedemann computation is estimated
in~\cite[Proposition~13]{cryptoeprint:2026/1339} by
\begin{equation}
c_{\mathrm{Wied}}(m,r,n,p)
\coloneqq N_1\left(
\underbrace{
N_1(n-mr)
\sum_{\ell=1}^{\lfloor\log_2p\rfloor}
\binom{p}{2^\ell}
}_{\substack{\text{sparse matrix--vector}\\\text{product cost}}}
+N_0\log(N_0)\right).
\label{eq:published-cost}
\end{equation}

The sparse matrix--vector product term in Hemmert's estimate counts
every potentially contributing product of binary coordinates as nonzero.
This is equivalent to
setting all involved binary coordinates to one; we call this the
\emph{all-one model}. The truncated zeta-transform evaluation changes
only this term. Since the sparse-matrix evaluation benefits from lower
column weights, we compare the two methods not only under the
all-one model but also under a \emph{Bernoulli$(1/2)$ model}, in which
these coordinates are modeled as i.i.d.\ Bernoulli$(1/2)$ random
variables.

\subsection{All-one model}

For the truncated zeta-transform evaluation, put $k=mr-s$ and define
\begin{equation}
 Z(k,p)\coloneqq \sum_{u=L+1}^{p}u\binom{k}{u},
 \qquad L=p-2^{\lfloor\log_2p\rfloor}.
 \label{eq:zeta-cost}
\end{equation}

If $a\in\F_2^k$ has Hamming weight
\[
w\coloneqq |\{i\in[k]:a_i=1\}|,
\]
then the number of active butterfly operations, namely those with
$a_i=1$, is exactly
\[
w\sum_{t=L}^{p-1}\binom{k-1}{t}.
\]
In the all-one model, $w=k$. Since
\[
 k\binom{k-1}{t}=(t+1)\binom{k}{t+1},
\]
the per-column operation count is
\[
 k\sum_{t=L}^{p-1}\binom{k-1}{t}
 =
 \sum_{u=L+1}^{p}u\binom{k}{u}
 =
 Z(k,p).
\]
Summing over the $n-mr$ contributing columns gives the following
Wiedemann cost estimate for the truncated zeta-transform evaluation:
\begin{equation}
 c_{\mathrm{Wied}}^{\mathrm{zeta}}(m,r,n,p)
 \coloneqq N_1\left(
 (n-mr)Z(mr-s,p)+N_0\log(N_0)
 \right).
 \label{eq:zeta-wiedemann}
\end{equation}

\subsection{Bernoulli$(1/2)$ model}

Suppose instead that the coordinates of each contributing binary column
are independent Bernoulli$(1/2)$ random variables. In the sparse-matrix
evaluation, an incidence at level $\ell$ requires $2^\ell$ specified
coordinates to be one and therefore occurs with probability
$2^{-2^\ell}$. In the truncated zeta-transform evaluation, butterfly
operations with $a_i=0$ may be omitted. Since $\mathbb E[w]=k/2$, the
expected per-column operation count is $Z(k,p)/2$. These expected
operation counts give the following Wiedemann cost estimates for the
sparse-matrix and truncated zeta-transform evaluations, respectively:
\begin{align}
 c_{\mathrm{Wied},1/2}(m,r,n,p)
 &\coloneqq
 N_1\left(
 N_1(n-mr)
 \sum_{\ell=1}^{\lfloor\log_2p\rfloor}
 2^{-2^\ell}\binom{p}{2^\ell}
 +N_0\log(N_0)
 \right),                                                   \label{eq:half-sparse}\\
 c_{\mathrm{Wied},1/2}^{\mathrm{zeta}}(m,r,n,p)
 &\coloneqq
 N_1\left(
 \tfrac12(n-mr)Z(mr-s,p)
 +N_0\log(N_0)
 \right).                                                   \label{eq:half-zeta}
\end{align}

Table~\ref{tab:work-estimates} evaluates the estimates for the
sparse-matrix and truncated zeta-transform evaluations under both
models for the Classic McEliece parameters and shortenings used
in~\cite[Table~5]{cryptoeprint:2026/1339}.

\begin{table}[!htbp]
\caption{Base-two logarithms of the Wiedemann work estimates.
``Reported'' reproduces the values
in~\cite[Table~5]{cryptoeprint:2026/1339}. ``Sparse'' denotes the
sparse-matrix evaluation, and ``Zeta'' denotes the truncated
zeta-transform evaluation proposed in this work.}
\label{tab:work-estimates}
\centering
\scriptsize
\resizebox{\textwidth}{!}{%
\begin{tabular}{lrrrrrrrrr}
\toprule
 & & & \multicolumn{1}{c}{\cite[Table~5]{cryptoeprint:2026/1339}} &
 \multicolumn{3}{c}{All-one} &
 \multicolumn{3}{c}{Bernoulli$(1/2)$} \\
\cmidrule(lr){4-4}\cmidrule(lr){5-7}\cmidrule(lr){8-10}
$(n,r,m)$ & $p$ & $s$ & Reported
& Sparse & Zeta & Gain
& Sparse & Zeta & Gain \\
\midrule
$(3488,64,12)$  & 22 & 150 & 297
& 297.02 & 282.93 & 14.09
& 289.19 & 281.94 &  7.25 \\
$(4608,96,13)$  & 46 & 158 & 594
& 593.77 & 559.21 & 34.56
& 577.55 & 558.22 & 19.33 \\
$(6688,128,13)$ & 52 & 256 & 690
& 693.12 & 651.93 & 41.19
& 673.43 & 650.94 & 22.48 \\
$(6960,119,13)$ & 41 & 275 & 564
& 563.74 & 532.56 & 31.18
& 548.04 & 531.57 & 16.47 \\
$(8192,128,13)$ & 39 & 321 & 550
& 549.04 & 519.24 & 29.80
& 533.55 & 518.25 & 15.30 \\
\bottomrule
\end{tabular}%
}
\end{table}

\FloatBarrier
\begingroup
\raggedright
\bibliographystyle{splncs04}
\bibliography{note_zeta_transform}
\endgroup

\end{document}